\documentclass[11pt]{article}

\usepackage{graphicx, verbatim}
\usepackage{amssymb}
\usepackage{alltt}
\usepackage{amsmath, amssymb, amsfonts, amscd, xspace, pifont, amsthm}
\usepackage{mathrsfs}
\usepackage{algorithm}
\newtheorem{theorem}{Theorem}

\newtheorem{proposition}{Proposition}
\newtheorem{corollary}{Corollary}

\usepackage[numbers]{natbib}
\usepackage{enumitem}   

\def\threeImages#1#2#3#4#5#6#7#8#9 
{
	\centerline{\hfill\makebox[#2]{#3}\hfill\makebox[#5]{#6}\hfill\makebox[#8]{#9}\hfill}
	\centerline{\hfill
		\includegraphics[width=#2]{#1}
		\hfill
		\includegraphics[width=#5]{#4}
		\hfill
		\includegraphics[width=#8]{#7}
		\hfill}
}

\def\twoImages#1#2#3#4#5#6 
{
	\centerline{\hfill\makebox[#2]{#3}\hfill\makebox[#5]{#6}\hfill}
	\centerline{\hfill
		\includegraphics[width=#2]{#1}
		\hfill
		\includegraphics[width=#5]{#4}
		\hfill}
}
\theoremstyle{remark}

\begin{document}
	\title{An Effective Bernstein-type Bound on Shannon Entropy over Countably Infinite Alphabets}
	\author{
		Yunpeng Zhao \thanks{School of Mathematical and Natural Sciences,
			Arizona State University, AZ, 85306. \texttt{Email:} yunpeng.zhao@asu.edu.} 
	} 
	
	\maketitle
	
	\begin{abstract}
		We prove a Bernstein-type bound for the difference between the average of negative log-likelihoods of independent discrete random variables and the Shannon entropy, both defined on a countably infinite alphabet. The result holds for the class of discrete random variables with tails  lighter than or on the same order of a discrete power-law distribution. Most commonly-used discrete distributions such as the Poisson distribution, the negative binomial distribution, and the power-law distribution itself belong to this class. The bound is effective in the sense that we provide a method to compute the constants in it. 
		
	\end{abstract}
	
	\section{Introduction}
	Concentration inequalities provide powerful tools for many subjects including  information theory \cite{raginsky2012concentration}, algorithm analysis \cite{dubhashi2009concentration} and statistics \cite{wainwright_2019,vershynin2018high}. 
	  The goal of the present paper is to prove an exponential decay bound with computable constants for the difference  between the negative log-likelihood of discrete random variables   and the Shannon entropy, both defined on  a countably infinite alphabet.  
	
	Let $X$ be a discrete random variable on a countably infinite alphabet $\mathcal{X}=\{x_1,...,x_k,... \}$. Let $p_k=\mathbb{P}(X=x_k)$ be the probability mass at $x_k$. Assume, without loss of generality, that $p_k>0$ for each $k$; otherwise, simply remove  $x_k$ with $p_k=0$ from $\mathcal{X}$.  Let $P(X)$ be the probability mass function, which is a random variable with $P(X)=p_k$ if $X=x_k$, $k\geq 1$. Then  $\mathbb{E}[-\log P(X)]=-\sum_{k=1}^\infty p_k \log p_k$ is the Shannon entropy\footnote{Throughout the paper, ``log'' denotes the natural logarithm.}, which is a key concept in information theory \cite{shannon1948mathematical,cover2006elements}. Note that neither $P(X)$ nor the entropy depends on the elements in $\mathcal{X}$. In fact, $\mathcal{X}$ is not necessarily a set of real numbers. The set  can contain generic symbols such as letters, and is therefore named as alphabet.  
	
	 Entropy on countably infinite alphabets does not always have finite values. We give a simple sufficient condition ensuring its finiteness  at the beginning of Section \ref{sec:main}, which is also the key assumption for the main result of the paper. The readers are referred to \cite{baccetti2013infinite} for a more thorough discussion on conditions for finiteness of entropy  on countably infinite alphabets. 
	
	Let $X_1,...,X_n$ be independently and identically distributed (i.i.d.) copies of $X$. Then $\sum_{i=1}^n \log P(X_i)$ is the joint log-likelihood of $X_1,...,X_n$. By the weak law of large numbers, 
	\begin{align*}
\mathbb{P} \left ( \left  |\frac{1}{n}\sum_{i=1}^n \log P(X_i) -\mathbb{E}[\log P(X)] \right  | \geq  \epsilon \right ) \rightarrow 0,
	\end{align*}
	provided that the entropy is finite. This  result, particularly for the case of $|\mathcal{X}|$ being finite, is called the asymptotic equipartition property in the information theory literature, which is the foundation of many important results in this field \cite{cover2006elements,csiszar2011information}.

	In this paper, we  strengthen the above result  by proving a Bernstein-type bound for the case of countably infinite alphabets:
	\begin{align}
\mathbb{P} \left ( \left  |\frac{1}{n}\sum_{i=1}^n \log P(X_i) -\mathbb{E}[\log P(X)] \right  | \geq  \epsilon \right ) \leq 2 \exp \left (- \frac{n\epsilon^2}{c_1+c_2\epsilon} \right ),  \label{motivating}
	\end{align}
	where $c_1$ and $c_2$ are computable constants that depend on $\{p_k\}_{k\geq 1}$. 
	
	Concentration inequalities for entropy have been studied recently. Zhao \cite{zhao2020note} proved a Bernstein-type inequality for entropy on finite alphabets with convergence rate $(K^2\log K)/n=o(1)$, where $n$ is the sample size and $K$ is the size of the alphabet. Zhao \cite{zhao2020optimal} proved an exponential decay bound that improves the rate  to $ (\log K)^2/n=o(1)$ and showed that the new rate is optimal. Both papers  studied  inequalities  for finite alphabets while we focus on countably infinite alphabets in this work. In Section \ref{sec:main},  we prove \eqref{motivating} under a mild assumption. In Section \ref{sec:condition}, we show that this assumption holds if  the tail of $\{p_k\}_{k\geq 1}$  drops faster or on the same order of a discrete power-law distribution; conversely, the assumption cannot be satisfied if the tail drops slower than any power-law distribution.  Most commonly-used discrete distributions such as the Poisson distribution, the negative binomial distribution, and the power-law distribution itself satisfy this assumption. Furthermore, we propose a method to compute the constants in the bound \eqref{motivating}. 
	
\section{Main Result}\label{sec:main}
Our result requires only one assumption on $\{p_k\}_{k\geq 1}$: 

\textbf{Assumption 1. } There exists $0<r<1$ such that 
\begin{align*}
\sum_{k=1}^\infty p_k^{1-r} \leq C_r  <\infty. 
\end{align*} 
Assumption 1 implies that the tail of $\{p_k\}_{k\geq 1}$ cannot be too heavy, and in Section \ref{sec:condition} we will elaborate this assumption by showing  that the assumption holds if the tail of $\{p_k\}_{k\geq 1}$ is lighter than or on the same order of  a discrete power-law distribution; conversely, it cannot be satisfied if the tail is heavier than any  power-law distribution. 

First note that Assumption 1 ensures the finiteness of the entropy.
\begin{proposition} \label{thm:finiteness}
	Under Assumption 1, $\mathbb{E}[-\log P(X)]<\infty.$
\end{proposition}
\begin{proof}
\begin{align*}
\mathbb{E}[-\log P(X)]=-\sum_{k=1}^\infty  p_{k} \log p_{k} \leq  \sum_{k=1}^\infty  p_k^{1-r} (-p_{k}^r \log p_{k} ) \leq \frac{1}{er}  \sum_{k=1}^\infty  p_k^{1-r}.
\end{align*}
The last inequality holds because $-p_{k}^r \log p_{k}$ on $[0,1]$ is maximized at $p_k=e^{-1/r}$. This result can be easily verified by comparing the function value at the  stationary point in $(0,1)$, which is unique for this function, with the values on the boundaries. Here we use the convention $q^r \log q=0$ at $q=0$, which makes the function  continuous on $[0,1]$ since $\lim_{q\rightarrow 0+} q^r \log q=0$.
\end{proof}
Let  
$
Y_i=\log P(X_i)-\mathbb{E}[\log P(X)]
$.
The key ingredient of the proof is to bound the moment generating function (MGF) of  $Y_i$, which is defined as
\begin{align*}
\mathbb{E}[e^{\lambda Y_i}]=  \left ( \sum_{k=1}^\infty p_{k}^{\lambda+1} \right ) \exp \left (-\lambda \sum_{k=1}^\infty p_{k} \log p_{k} \right ). 
\end{align*}
Denote the MGF of $Y_i$ by  $M_{Y_i}(\lambda)$. Under Assumption 1, $M_{Y_i}(\lambda)$ is finite for $|\lambda|<r$ because
\begin{align*}
\sum_{k=1}^\infty p_{k}^{\lambda+1} \leq \sum_{k=1}^\infty p_k^{1-r} < \infty.
\end{align*}
Conversely, if Assumption 1 does not hold then  $\sum_{k=1}^\infty p_{k}^{\lambda+1}$ diverges for all $\lambda<0$, because if $\sum_{k=1}^\infty p_{k}^{\lambda+1}$ converges for a certain negative $\lambda$  then it must be in the interval $(-1,0)$ and one can  take $r=-\lambda$. 

We now give the main result. 
\begin{theorem}[Main result]\label{thm:main}
	Under Assumption 1, that is, if there exists  $0<r<1$ such that 
	\begin{align*}
	\sum_{k=1}^\infty p_k^{1-r} \leq C_r  <\infty,
	\end{align*} 
	then for $|\lambda|< r$,
	\begin{align*}
	 M_{Y_i}(\lambda) \leq \exp \left (  \frac{ C_r \lambda^2}{r^2}\frac{1}{1-\frac{|\lambda|}{r}}  \frac{1}{2\sqrt{\pi}} \right ).
	\end{align*}
	Furthermore, for all $\epsilon>0$, 
	\begin{align}
	\mathbb{P} \left ( \left  |\frac{1}{n}\sum_{i=1}^n \log P(X_i) -\mathbb{E}[\log P(X)]  \right  | \geq  \epsilon \right ) \leq  2 \exp \left ( - \frac{n \epsilon^2}{2  C_r/(\sqrt{\pi}r^2)+2 \epsilon /r}  \right ). \label{core}
	\end{align}
\end{theorem}
\begin{proof}
For $|\lambda|< r$, 
	\begin{align}
\log  M_{Y_i}(\lambda)  = & \,\,  \log \left (  \sum_{k=1}^\infty p_{k}^{\lambda+1} \right )-\lambda \sum_{k=1}^\infty p_{k} \log p_{k} \nonumber \\
\leq & \,\,  \sum_{k=1}^\infty p_{k}^{\lambda+1}-1-\lambda \sum_{k=1}^\infty p_{k} \log p_{k} \nonumber  \\
 = & \,\,  \sum_{k=1}^\infty p_{k} \exp(\lambda \log p_{k})-1-\lambda \sum_{k=1}^\infty p_{k} \log p_{k} \nonumber \\
 =  & \,\,\sum_{k=1}^\infty \left ( p_k+ \lambda p_{k} \log p_{k}+ \sum_{m=2}^\infty \frac{1}{m!} \lambda^m p_k (\log p_k)^m  \right ) -1-\lambda \sum_{k=1}^\infty p_{k} \log p_{k}, \label{log_mgf}
\end{align} 
where the  inequality follows from $\log x \leq x-1$ for $x>0$. 

For $m\geq 2$, it is easy to check that, the minimum of $p_{k}^r (\log p_{k})^m$ on $[0,1]$ when $m$ is an odd number, and the maximum  when $m$ is an even number,  are achieved at $e^{-m/r}$ by comparing the function value at the stationary point in $(0,1)$, which is unique,  with the values on the boundaries. Here we use the convention $q^r (\log q)^m=0$ at $q=0$ as before, which makes the function  continuous on $[0,1]$ since $\lim_{q\rightarrow 0+} q^r (\log q)^m=0$.

Therefore, for $m\geq 2$, 
\begin{align}
&  \,\, \left | \frac{1}{m!} \lambda^m p_k (\log p_k)^m  \right | \nonumber \\
\leq &  \,\, p_k^{1-r} \frac{1}{m!} |\lambda|^m  | p_{k}^r (\log p_{k})^m | \nonumber \\
\leq & \,\,  p_k^{1-r} \frac{1}{m!} |\lambda|^m e^{-m} \left ( \frac{m}{r} \right )^m \nonumber  \\
= & \,\, p_k^{1-r}  \frac{1}{m!} \left ( \frac{|\lambda|}{r} \right )^m  \left ( \frac{m}{e} \right )^m \nonumber \\
\leq &  \,\, p_k^{1-r} \frac{1}{m!}  \left ( \frac{|\lambda|}{r} \right )^m  \frac{m!}{\sqrt{2\pi m}} \nonumber \\
\leq & \,\, p_k^{1-r}  \left ( \frac{|\lambda|}{r} \right )^m  \frac{1}{2\sqrt{\pi}}, \label{radius}
\end{align}
where the second inequality is obtained by replacing $| p_{k}^r (\log p_{k})^m |$ with its maximum and the third inequality follows from Stirling's formula (see \cite{robbins1955remark} for example):
\begin{align*}
m! \geq \sqrt{2 \pi m}  \left ( \frac{m}{e} \right )^m , \,\, \textnormal{for } m\geq 1.
\end{align*} 

It follows that for $|\lambda|<r$, 
\begin{align*}
 \left | \sum_{m=2}^\infty \frac{1}{m!} \lambda^m p_k (\log p_k)^m \right |  \leq \sum_{m=2}^\infty \left | \frac{1}{m!} \lambda^m p_k (\log p_k)^m  \right |   \leq   p_k^{1-r}\sum_{m=2}^\infty   \left ( \frac{|\lambda|}{r} \right )^m \frac{1}{2\sqrt{\pi}} =   p_k^{1-r} \frac{\lambda^2}{r^2}\frac{1}{1-\frac{|\lambda|}{r}} \frac{1}{2\sqrt{\pi}},
\end{align*}
and
\begin{align*}
\sum_{k=1}^\infty \left | \sum_{m=2}^\infty \frac{1}{m!} \lambda^m p_k (\log p_k)^m \right | \leq C_r \frac{\lambda^2}{r^2}\frac{1}{1-\frac{|\lambda|}{r}}  \frac{1}{2\sqrt{\pi}}.
\end{align*}

Since the three terms  under $\sum_{k=1}^\infty$ in \eqref{log_mgf} all converge absolutely for $|\lambda|<r$, one can take the sum term by term. Therefore, for $|\lambda|<r$, 
\begin{align*}
\log  M_{Y_i}(\lambda) \leq \sum_{k=1}^\infty \sum_{m=2}^\infty \frac{1}{m!} \lambda^m p_k (\log p_k)^m \leq  C_r \frac{\lambda^2}{r^2}\frac{1}{1-\frac{|\lambda|}{r}} \frac{1}{2\sqrt{\pi}}, 
\end{align*}
and
\begin{align}
	 M_{Y_i}(\lambda) \leq \exp \left (  \frac{ C_r \lambda^2}{r^2}\frac{1}{1-\frac{|\lambda|}{r}}  \frac{1}{2\sqrt{\pi}} \right ). \label{mgf_bound}
\end{align}

The second part of the theorem  follows from a standard argument using the Chernoff bound, which can be found in Chapter 2 of \cite{wainwright_2019}. We give the details for completeness. For $t>0$ and $0<\lambda<r$,
\begin{align*}
\mathbb{P} \left (     \sum_{i=1}^n Y_i  \geq t  \right ) = \mathbb{P} \left ( e^{\lambda \sum_{i=1}^n Y_i}  \geq e^{\lambda t}  \right ) \leq \frac{\prod_{i=1}^n M_{Y_i}(\lambda)}{ e^{\lambda t} } \leq \exp \left \{  \frac{ n C_r \lambda^2}{r^2}\frac{1}{1-\frac{|\lambda|}{r}}  \frac{1}{2\sqrt{\pi}}-\lambda t \right  \}, \\
\end{align*}
where the first inequality is Markov's inequality and the second inequality follows from \eqref{mgf_bound}. By setting 
\begin{align*}
\lambda=\frac{t}{ n C_r/(\sqrt{\pi}r^2)+t/r} \in (0,r),
\end{align*}
we obtain
\begin{align*}
\mathbb{P} \left (     \sum_{i=1}^n Y_i  \geq t  \right ) \leq \exp \left ( - \frac{t^2}{2 n C_r/(\sqrt{\pi}r^2)+2t/r}  \right ). 
\end{align*}
The left tail bound can be obtained similarly by setting $\lambda=-\frac{t}{ n C_r/(\sqrt{\pi}r^2)+t/r}$. Therefore,
\begin{align*}
\mathbb{P} \left ( \left |\sum_{i=1}^n {Y_i} \right | \geq t \right ) \leq  2 \exp \left ( - \frac{t^2}{2 n C_r/(\sqrt{\pi}r^2)+2t/r}  \right ).
\end{align*}
Finally, letting $t=n\epsilon$,
\begin{align*}
\mathbb{P} \left ( \left |\frac{1}{n} \sum_{i=1}^n  {Y_i} \right | \geq \epsilon \right ) \leq 2 \exp \left ( - \frac{n \epsilon^2}{2  C_r/(\sqrt{\pi}r^2)+2 \epsilon /r}  \right ).
\end{align*}
\end{proof}

Theorem \ref{thm:main} can be generalized to $\{X_i\}_{i=1,...,n}$ with independent but non-identical distributions. Let $p_{ik}=\mathbb{P}(X_i=x_k)$ be the probability mass of $X_{i}$ at $x_k$ and  $\mathbb{E}[-\log P(X_i)]=-\sum_{k=1}^\infty p_{ik} \log p_{ik}$ be the entropy of $X_i$. Furthermore, redefine $Y_i$ and $M_{Y_i}(\lambda)$  accordingly. We have the following result for non-identical distributions: 
\begin{corollary}
If there exists  $0<r<1$ such that 
\begin{align*}
\sum_{k=1}^\infty p_{ik}^{1-r} \leq C_{r,i}  <\infty, \,\, i=1,...,n, 
\end{align*} 
	then for $|\lambda|< r$,
\begin{align*}
M_{Y_i}(\lambda) \leq \exp \left (  \frac{ C_{r,i} \lambda^2}{r^2}\frac{1}{1-\frac{|\lambda|}{r}}  \frac{1}{2\sqrt{\pi}} \right ).
\end{align*}
Furthermore, for all $\epsilon>0$, 
\begin{align*}
\mathbb{P} \left ( \left  |\frac{1}{n}\sum_{i=1}^n  \left ( \log P(X_i) -\mathbb{E}[\log P(X_i)] \right ) \right  | \geq  \epsilon \right ) \leq  2 \exp \left ( - \frac{n \epsilon^2}{2  \sum_{i=1}^n C_{r,i}/(n \sqrt{\pi}r^2)+2 \epsilon /r}  \right ).
\end{align*}
\end{corollary}
The proof is the same as of Theorem \ref{thm:main}.

\section{Determining the Constants in the Bound}\label{sec:condition}
The radius of convergence $r$ in \eqref{radius} and  the upper bound $C_r$  for $\sum_{k=1}^\infty p_k^{1-r}$ are the only constants to be determined if one wants to use \eqref{core} as an effective upper bound for a given distribution $\{p_k\}_{k\geq 1}$. 

We first determine the types of distributions and the range of $r$ that can make $\sum_{k=1}^\infty p_k^{1-r}$ converge. Intuitively speaking,  for  distributions that satisfy Assumption 1, the tail of $\{p_k\}_{k\geq 1}$ cannot be too heavy.
 We make the above statement precise in the following theorem.
\begin{theorem}\label{thm:power}
The distribution $\{p_k\}_{k\geq 1}$ satisfies  Assumption 1 if the tail of $\{p_k\}_{k\geq 1}$ is lighter than or on the same order of  a discrete power-law distribution; conversely, Assumption 1 cannot be satisfied if the tail is heavier than any  power-law distribution. Specifically,
	\begin{enumerate}[label=(\roman*)]
		\item If 
		\begin{align*}
		\lim_{k \rightarrow \infty} \frac{p_k}{k^{-\alpha}} =0, \,\, \textnormal{for all } \alpha>1, 
		\end{align*}
		then 
		\begin{align*}
		\sum_{k=1}^\infty p_k^{1-r}  <\infty,\,\, \textnormal{for all } 0<r<1.
		\end{align*}
		\item If
		\begin{align*}
		0< \liminf_{k \rightarrow \infty} \frac{p_k}{k^{-\alpha}} \leq  \limsup_{k \rightarrow \infty} \frac{p_k}{k^{-\alpha}}<\infty, \,\, \textnormal{for some } \alpha>1,  
		\end{align*}
		then 
		\begin{align*}
		\sum_{k=1}^\infty p_k^{1-r}  <\infty, \,\, \textnormal{if and only if } 0 <r<\frac{\alpha-1}{\alpha}.
		\end{align*}
		\item If
		\begin{align*}
		\lim_{k \rightarrow \infty} \frac{p_k}{k^{-\alpha}}=\infty, \,\, \textnormal{for all } \alpha>1, 
		\end{align*}
		then 
		\begin{align*}
		\sum_{k=1}^\infty p_k^{1-r}  = \infty, \,\, \textnormal{for all } 0<r<1.
		\end{align*}
	\end{enumerate}
\end{theorem}
\begin{proof}
	Recall that  $\sum_{k=1}^\infty k^{-\beta}$ converges for $\beta>1$, and diverges for $\beta\leq 1$. Statement (\romannumeral 1) is obvious by taking $\alpha>1/(1-r)$.  Statement (\romannumeral 2) is also obvious by noticing that the assumption implies that there exist positive constants $a_1,a_2$ such that $ a_1 k^{-\alpha} \leq p_k\leq a_2 k^{-\alpha}$ for sufficiently large $k$.   We prove (\romannumeral 3) by contradiction. If there exists $0<r<1$ such that $\sum_{k=1}^\infty p_k^{1-r}  < \infty$, then 
	\begin{align*}
	\liminf_{k \rightarrow \infty} \frac{p_k^{1-r}}{k^{-1}} =0.
	\end{align*}
	It implies
	\begin{align*}
	\liminf_{k \rightarrow \infty} \frac{p_k}{k^{-1/(1-r)}} =0,
	\end{align*}
	which contradicts the assumption since $1/(1-r)>1$. 
\end{proof}

Theorem \ref{thm:power} implies that there are a wide class of discrete distributions satisfying Assumption 1,  including the most commonly-used ones such as the Poisson distribution, the negative binomial distribution, and the power-law distribution itself. The class even contains certain discrete random variables that do not have finite expectations. In fact, if $X$ follows a discrete power-law distribution with $1<\alpha \leq 2$ then  $\mathbb{E}[X]=\infty$ since $\sum_{k=1}^\infty  k^{-(\alpha-1)}$ diverges. But such distributions satisfy Assumption 1 by Theorem \ref{thm:power} (\romannumeral 2).

\noindent \textit{Remark.} It may be surprising, at first glance, to get an exponential decay bound for a power-law distribution, which itself is heavy-tailed. But note that \eqref{core} is a  concentration bound for $\log P(X)$, not for $X$. The log-likelihood $\log P(X)$ is typically better-behaved than $X$ that takes values on non-negative integers and follows a heavy-tailed distribution. For example, for a power-law distribution with $1<\alpha \leq 2$, $\mathbb{E}[X]=\infty$; on the contrary, the entropy $\mathbb{E}[-\log P(X)]$ is finite by Proposition \ref{thm:finiteness} and Theorem \ref{thm:power} (\romannumeral 2).  This phenomenon can be explained by noticing that $-\log (k^{-\alpha})$ grows much slower than $k$. Moreover, the MGF of $X$ is infinite if $X$ follows a power-law distribution while the MGF of $\log P(X)$ is finite. The tail of $\log P(X)$ is not heavy in this sense, which makes \eqref{core} possible. 

Finally, we discuss how to compute $C_r$ after $r$ is selected  by Theorem \ref{thm:power}. In practice, one can compute the partial sum of $\sum_{k=1}^\infty p_k^{1-r}$ until the increment is negligible. The value obtained in this way, however, is a lower bound for $\sum_{k=1}^\infty p_k^{1-r}$ and a generic truncation error bound does not exist for positive infinite series because in principle, the  tail behavior cannot be predicted by a finite number of terms\footnote{This issue is minor in practice especially when $p_k$ drops exponentially. The series $\sum_{k=1}^\infty p_k^{1-r}$  usually converges very fast in this case. It is nothing wrong to take the partial sum until the increment is negligible. The method in Theorem \ref{thm:c_r} is useful to someone who needs a rigorous upper bound.}. 

 If the tail of $\{p_k \}_{k\geq 1}$ is dominated by a power-law distribution, we propose a method that can compute an upper bound for $\sum_{k=1}^\infty p_k^{1-r}$ at any tolerance level. Specifically, the next theorem shows how to compute an upper bound $C_r$ for $\sum_{k=1}^\infty p_k^{1-r}$ with $|\sum_{k=1}^\infty p_k^{1-r}-C_r|$ smaller than a pre-specified tolerance level if we find $k_0$ such that $p_k \leq c_0 k^{-\alpha}$ for $k>k_0$. Note that such $k_0$ exists  if $\{p_k\}_{k\geq 1}$ satisfies the condition in (\romannumeral 1) or (\romannumeral 2) in Theorem \ref{thm:power}.

\begin{theorem}\label{thm:c_r}
	Suppose $k_0$ is a positive integer such that $p_k \leq c_0 k^{-\alpha}$  for certain $\alpha>1$ and all $k>k_0$, where $c_0>0$. Pick $r$ such that $0<r< (\alpha-1)/\alpha$. For all $\epsilon>0$, let 
	\begin{align*}
	k_1=\max \left \{ k_0, \left \lceil \left ( \frac{\epsilon(\alpha(1-r)-1)}{c_0} \right )^{-\frac{1}{\alpha(1-r)-1}} \right \rceil \right \},
	\end{align*}
	 where $\lceil \cdot \rceil$ means rounding up to the next integer.
	Then 
	\begin{align*}
	C_r=\sum_{k=1}^{k_1} p_k^{1-r}+\epsilon
	\end{align*} 
	satisfies
	\begin{align*}
	0 \leq  C_r-\sum_{k=1}^\infty p_k^{1-r}  \leq \epsilon. 
	\end{align*}
\end{theorem}
\begin{proof} We only need to bound the tail probability for $k>k_1$.  
		\begin{align*}
		\sum_{k=k_1+1}^\infty  p_k^{1-r} & \leq  c_0 \sum_{k=k_1+1}^\infty k^{-\alpha(1-r)} \\
		& = c_0 \sum_{k=k_1}^\infty \int_{k}^{k+1} (k+1)^{-\alpha(1-r)} \,\, dx \\
		& \leq c_0 \int_{k_1}^\infty x^{-\alpha(1-r)} \,\, dx \\
		& = \frac{c_0}{\alpha(1-r)-1} k_1^{-(\alpha(1-r)-1)} \leq \epsilon,  \\
		\end{align*}
	where the  first inequality holds because $p_k \leq c_0 k^{-\alpha}$ for all $k>k_0$ and the last inequality holds because $k_1\geq \left \lceil \left ( \frac{\epsilon(\alpha(1-r)-1)}{c_0} \right )^{-\frac{1}{\alpha(1-r)-1}} \right \rceil$. 
	
	Therefore, 
	\begin{align*}
	\sum_{k=1}^\infty p_k^{1-r} & = \sum_{k=1}^{k_1} p_k^{1-r}+\sum_{k=k_1+1}^\infty  p_k^{1-r} \leq \sum_{k=1}^{k_1} p_k^{1-r}+\epsilon. 
	\end{align*}
\end{proof}	
	
	\section*{Acknowledgements}
	This research was supported by the National Science Foundation grant DMS-1840203.
	
 \newcommand{\noop}[1]{}

\end{document}